\documentclass[12pt]{amsart}

\usepackage{amsmath,amsthm,amssymb,amsfonts,verbatim}
\usepackage[colorlinks,
            linkcolor=blue,
            anchorcolor=blue,
            citecolor=blue
            ]{hyperref}
\usepackage[hmargin=1.2in,vmargin=1.2in]{geometry}

\title[LRC via automorphism group]{Construction of optimal locally repairable codes via automorphism groups of rational function fields}
\author{Lingfei Jin}\address{Shanghai Key Laboratory of Intelligent Information Processing, School of Computer Science, Fudan University, Shanghai 200433, China.} \email{lfjin@fudan.edu.cn}
\author{Liming Ma}\address{School of Mathematical Sciences, Yangzhou University, Yangzhou China
225002}\email{lmma@yzu.edu.cn}
\author{Chaoping Xing} \address{Division of Mathematical Sciences, School of Physical \& Mathematical Sciences,
Nanyang Technological University, Singapore
637371}\email{xingcp@ntu.edu.sg}
\date{}

\newtheorem{lemma}{Lemma}[section]
\newtheorem{theorem}[lemma]{Theorem}
\newtheorem{cor}[lemma]{Corollary}

\newtheorem{prop}[lemma]{Proposition}
\newtheorem{ex}[lemma]{Example}

\newtheorem{defn}{Definition}

\theoremstyle{remark}
\newtheorem{rmk}{Remark}

\renewcommand{\epsilon}{\varepsilon}
\renewcommand{\le}{\leqslant}
\renewcommand{\ge}{\geqslant}

\def\PGL{{\rm PGL}}
\def\Gal{{\rm Gal}}

\newcommand{\vnote}[1]{}

\def\ZZ{\mathbb{Z}}
\def\PP{\mathbb{P}}

\def\F{\mathbb{F}}

\def \mA {\mathcal{A}}

\def \mA {\mathcal{A}}

\def \mL {\mathcal{L}}

\def \mP {\mathcal{P}}

\def \Xi {{X^{[i]}}}

\newcommand{\Ga}{\alpha}
\newcommand{\Gb}{\beta}
     
\newcommand{\Gd}{\delta}     
\newcommand{\Ge}{\epsilon}

    \newcommand{\GL}{\Lambda}
     
\newcommand{\Gs}{\sigma}     
\newcommand{\Gt}{\tau}

\def\Tr{{\rm Tr}}
\def\GL{{\rm GL}}
\def\PGL{{\rm PGL}}

\def \bc {{\bf c}}

\def\supp {{\rm supp }}

\def\mG{{\mathcal G}}

\def\Aut {{\rm Aut }}

\def\AGL {{\rm AGL}}
\def\PSL {{\rm PSL}}
\def\LRC {{\rm locally repairable code\ }}
\def\LRCs {{\rm locally repairable codes\ }}

\def\Gal{{\rm Gal}}

\begin{document}

\maketitle

\begin{abstract}
Locally repairable codes, or locally recoverable codes (LRC for short) are designed for application in distributed and cloud storage systems. Similar to classical block codes, there is an important bound called the Singleton-type bound for locally repairable codes. In this paper, an optimal locally repairable code refers to a block code achieving this Singleton-type bound. Like classical MDS codes, optimal locally repairable codes carry some very nice combinatorial structures. Since introduction of the Singleton-type bound for locally repairable codes, people have put tremendous effort on constructions of optimal locally repairable codes. Due to hardness of this problem, there are few constructions of optimal locally repairable codes in literature. Most of these constructions are realized via either combinatorial or algebraic structures.

In this paper, we employ automorphism groups of rational function fields to construct  optimal locally repairable codes by considering the group action on the projective lines over finite fields. It turns out that we are able to construct optimal locally repairable codes with reflexibility of locality as well as smaller alphabet size  comparable to the code length. In particular, we produce new families of $q$-ary locally repairable codes, including codes of length $q+1$ via  cyclic groups and codes via dihedral groups.
\end{abstract}

\section{Introduction}
Due to recent applications to distributed and cloud storage systems, a new class of block codes, i.e., locally repairable codes have been investigated by many researchers \cite{BTV17,HL07,HCL,GHSY12,PKLK12,CM13,SRKV13,FY14,PD14,TB14,TPD16}. A locally repairable code is just a  block code with an additional property called {\it locality}. As most of  classical block codes do not carry good locality, people have to investigate various constructions of  block codes with good locality.

Block codes with good locality were initially studied in \cite{HL07,HCL}, although \cite{HCL} considered a slightly different definition of locality, under which a code is said to have information locality.
Codes with information locality property were also studied in \cite{GHSY12,FY14}. For a locally repairable code $C$ of length $n$ with $k$ information symbols and locality $r$ (see the definition of locally repairable codes in Section \ref{subsec:2.1}), it was proved in \cite{GHSY12} that the minimum distance $d(C)$ of $C$ is upper bounded by
 \begin{equation}\label{eq:x1}
 d(C)\le n-k-\left\lceil \frac kr\right\rceil+2.
 \end{equation}
The bound \eqref{eq:x1} is called the Singleton-type bound for locally repairable codes and was proved by extending the arguments in the proof of the classical Singleton bound on codes. Like the classical Singleton bound,  the Singleton-type bound \eqref{eq:x1} does not take into account the cardinality of the code alphabet size $q$. Augmenting this result, a recent work \cite{CM13} established a bound on the distance of locally repairable codes that depends on $q$, sometimes yielding better results. A locally repairable code achieving any of these bound is optimal. However,  in this paper, we specifically refer an optimal locally repairable code to a block code achieving the bound \eqref{eq:x1}.

\subsection{Known results}
Apparently, construction of optimal  locally repairable codes, i.e., block codes achieving the bound \eqref{eq:x1} is of both theoretical interest and practical importance. However, due to hardness of this problem, there are few constructions available in the literature and hence few classes of  optimal  locally repairable codes are known. A class of codes constructed earlier and known as pyramid codes \cite{HCL} were shown to be  optimal locally repairable codes.  In \cite{SRKV13},  Silberstein {\it et al.}  proposed a two-level construction based on the Gabidulin codes combined with a single parity-check $(r+1,r)$ code. Another construction in \cite{TPD16} used two layers of MDS codes, a Reed-Solomon code and a special $(r+1,r)$ MDS code. A common shortcoming of these constructions relates to the size of the code alphabet which in all the papers is an exponential function of the code length, complicating the implementation. There  was an earlier construction of optimal locally repairable codes given in \cite{PKLK12} with  alphabet  size comparable to code length. However, the construction in \cite{PKLK12} only produced  a specific value of the length $n$, i.e., $n=\left\lceil \frac kr\right\rceil(r+1)$. Thus, the rate of the code is very close to $1$.

A recent breakthrough construction was given in \cite{TB14}. This construction naturally generalizes the Reed-Solomon construction which relies on the alphabet of cardinality comparable to the code length $n$. The idea behind the construction is very nice. The only shortcoming of this construction is the restriction on locality $r$. Namely,  $r+1$ must be a divisor of either $q-1$ or $q$, or $r+1$ is equal to some special divisors of $q(q-1)$.

There are also some existence results given in \cite{PKLK12} and \cite{TB14} with less restriction on locality $r$. But both results require large alphabet size which is an exponential function of the code length.

\subsection{Our results and comparison}
In this paper, we present a non-trivial generalization of optimal  locally repairable codes given in \cite{TB14} by employing automorphism groups of rational function fields.  This allows reflexibility of locality as well as smaller alphabet of cardinality comparable to the code length $n$. More precisely, as long as there is a subgroup of size $r+1$ in the projective general linear group $\PGL_2(q)$ (see the definition of $\PGL_2(q)$ in Section \ref{subsec:2.3}) with $(r+1)|n$ and $n\le q+1$, we can construct an optimal locally repairable code of length $n$ and locality $r$. Thus, to construct optimal  locally repairable codes, we need to find all subgroups of $\PGL_2(q)$. Fortunately, subgroups of $\PGL_2(q)$ have been completely determined in literature (see Section \ref{subsec:2.3}). Thus, we are able to obtain $q$-ary  optimal locally repairable codes of length $n$ and locality $r$ as long as there is a subgroup of size $r+1$ in the projective general linear group $\PGL_2(q)$  together with the conditions  $(r+1)|n$ and $n\le q+1$.

However, it is unnecessary to write down all subgroups and provide the corresponding constructions of  optimal locally repairable codes. Instead, we present a general construction for all subgroups of $\PGL_2(q)$ in Section \ref{subsec:2.3}. We then present explicit constructions of  optimal locally repairable codes from subgroups of the affine linear group $\AGL_2(q)$ (see definition of $\AGL_2(q)$ in Section \ref{subsec:2.2}). It turns out that the construction given in \cite{TB14} can be realized by subgroups of  $\AGL_2(q)$ under  our framework.

In addition, we also present an explicit construction from subgroups of size that divides $q+1$. This construction produces optimal locally repairable codes of length $q+1$ that were not known before. Then in the last section, we give optimal locally repairable codes from subgroups of a dihedral group that can be viewed as a subgroup of $\PGL_2(q)$.

\subsection{Organization of the paper}
In Section \ref{sec:2}, we introduce some backgrounds for this paper such as definition of locally repairable codes, automorphism groups of rational function fields, projective general linear groups and their subgroups, Hilbert's ramification  theory and algebraic geometry codes. Section \ref{sec:3} devotes to a general construction of  optimal locally repairable codes from arbitrary subgroups of $\PGL_2(q)$. In Section \ref{sec:4}, we present explicit constructions of optimal locally repairable codes from  subgroups of $\AGL_2(q)$. In Sections 5 and 6, we give explicit constructions of  optimal locally repairable codes from subgroups of a cyclic group of order $q+1$ and dihedral groups, respectively. In the last section, we conclude the paper by summarizing our main results of this paper.

\section{Preliminaries}\label{sec:2}
In this section, we present some preliminaries on locally repairable codes, automorphism groups of rational function fields, the subgroups of projective general linear groups $\PGL_2(q)$, Hilbert's ramification theory and algebraic geometry codes.

\subsection{Locally repairable codes}\label{subsec:2.1}
Informally speaking, a block code is said with locality $r$ if  every coordinate of a given codeword can be recovered by accessing at most $r$ other coordinates of this codeword. The formal definition of a locally repairable code with locality $r$ is given as follows.

\begin{defn}
Let $C\subseteq \F_q^n$ be a $q$-ary block code of length $n$. For each $\Ga\in\F_q$ and $i\in \{1,2,\cdots, n\}$, define $C(i,\Ga):=\{\bc=(c_1,\dots,c_n)\in C\; | \; c_i=\Ga\}$. For a subset $I\subseteq \{1,2,\cdots, n\}\setminus \{i\}$, we denote by $C_{I}(i,\Ga)$ the projection of $C(i,\Ga)$ on $I$.
Then $C$ is called a locally repairable code with locality $r$ if, for every $i\in \{1,2,\cdots, n\}$, there exists a subset
$I_i\subseteq \{1,2,\cdots, n\}\setminus \{i\}$ with $|I_i|\le r$ such that  $C_{I_i}(i,\Ga)$ and $C_{I_i}(i,\Gb)$ are disjoint for any $\Ga\neq \Gb$.
\end{defn}
Apart from the usual parameters: length, rate and minimum distance,  the locality of a   locally repairable code plays a crucial role. In this paper, we always consider locally repairable codes that are linear over $\F_q$. Thus, a $q$-ary \LRC of length $n$, dimension $k$, minimum distance $d$ and locality $r$ is said to be an $[n,k,d]_q$-\LRC with locality $r$.

If we ignore the minimum distance of a $q$-ary locally repairable code, then there is a constraint on the rate \cite{GHSY12}, namely,
\begin{equation}\label{eq:x2}
\frac{k}{n} \le \frac{r}{r+1}.
\end{equation}

In this paper, the minimum distance of a \LRC is taken into consideration and we always refer to the bound \eqref{eq:x1}. For an $[n,k,d]$-linear code, $k$ information symbols can recover the whole codeword. Thus, the locality $r$ is usually upper bounded by $k$. If we allow $r=k$, i.e., there is no constraint on locality, then the bound \eqref{eq:x1} becomes the usual Singleton bound that shows constraint on $n,k$ and $d$ only. The other extreme case is that the locality $r$ is $1$. In this case, the \LRC is a repetition code by repeating each symbol twice and the bound \eqref{eq:x1} becomes $d(C)\le n-2k+2$ which shows  the Singleton bound for repetition codes.

\subsection{Rational function fields and their automorphism groups}\label{subsec:2.2}
Let us introduce some basic facts of rational function fields over finite fields. The reader may refer to \cite{C51,St09} for the details.

Let $F$ be the rational function field $\F_q(x)$ over $\F_q$, where $x$ is transcendental over $\F_q$. First of all, a well-known fact is that every subfield $E$ of $F$ with $\F_q\varsubsetneq E\subseteq F$ is again a rational function field, i.e., $E=\F_q(z)$ for some element $z\in F$ that is transcendental over $\F_q$ (see \cite[Proposition 3.5.9]{St09}).

For every irreducible polynomial $P(x)\in \F_q[x]$, we define a discrete valuation $\nu_P$ which is a map   from $\F_q[x]$ to $\ZZ\cup\{\infty\}$ given by $\nu_P(0)=\infty$ and $\nu_P(f)=a$, where $f$ is a nonzero polynomial and $a$ is the unique nonnegative integer satisfying $P^a|f$ but
$P^{a+1}\nmid f$. This map can be extended to $\F_q(x)$ by defining $\nu_P(f/g)=\nu_P(f)-\nu_P(g)$ for any two polynomials $f, g\in\F_q[x]$ with $g\neq0$.
Apart from the above finite discrete valuation $\nu_P$, we have an infinite valuation $\nu_{\infty}$ (or $\nu_{P_\infty}$) defined by $\nu_{\infty}(f/g)=\deg(g)-\deg(f)$ for any two polynomials $f, g\in\F_q[x]$ with $g\neq0$. Note that we define $\deg(0)=\infty$. The set of places of $F$ is denoted by $\PP_F$.

For each discrete valuation $\nu_P$ ($P$ is either a polynomial or $P_\infty=\infty$), by abuse of notation we still denote by $P$ the set $\{y\in F:\; \nu_P(y)>0\}$. Then the set $P$ is called a place of $F$.
If $P=x-\Ga$, then we denote $P$ by $P_{\Ga}$. The degree of the place $P$ is defined to be the degree of the corresponding polynomial $P(x)$. If $P$ is the infinite place $\infty$, then the degree of $\infty$ is defined to be $1$. A place of degree $1$ is called rational. In fact, there are exactly $q+1$ rational places for the rational function field $F$ over $\F_q$. Moreover, the genus of $F$ is $0$. For a nonzero function $z$ of $F$, the zero divisor of $z$ is defined to be $(z)_0:=\sum_{P\in \PP_F, \nu_P(z)>0}\nu_P(z)P$. Similarly, the pole divisor of $z$ is defined to be $(z)_{\infty}:=-\sum_{P\in \PP_F, \nu_P(z)<0}\nu_P(z)P$.

We denote by  $\Aut(F/\F_q)$ the automorphism group of $F$ over $\F_q$, i.e.,
\begin{equation}
\Aut(F/\F_q)=\{\Gs: F\rightarrow F |\; \Gs  \mbox{ is an } \F_q\mbox{-automorphism of } F\}.
\end{equation}
It is well known that any automorphism $\Gs\in \Aut(F/\F_q)$ is uniquely determined by $\Gs(x)$
and given by
\begin{equation}\label{abcd}
\Gs(x)=\frac{ax+b}{cx+d}
\end{equation}
for some constants $a,b,c,d\in\F_q$ with $ad-bc\neq0$ (see \cite{C51}).
Denote by  $\GL_2(q)$ the general linear group of $2\times 2$ invertible matrices over $\F_q$.
Thus, every matrix $A=\left(\begin{array}{cc}a&b\\ c&d\end{array}\right)\in \GL_2(q)$ induces an automorphism of $F$ given by \eqref{abcd}.
Two matrices of $\GL_2(q)$ induce the same automorphism of $F$ if and only if they belong to the same coset of $Z(\GL_2(q))$, where $Z(\GL_2(q))$ stands for the center $\{aI_2:\; a\in\F_q^*\}$ of $\GL_2(q)$.  This implies that $\Aut(F/\F_q)$ is isomorphic to the projective general linear group $\PGL_2(q):=\GL_2(q)/Z(\GL_2(q))$. Thus, we can identify $\Aut(F/\F_q)$ with $\PGL_2(q)$.

Consider a subgroup of $\PGL_2(q)$:
\begin{equation}\label{eq:x3}
\AGL_2(q):=\left\{\left(\begin{array}{cc}a&b\\ 0&1\end{array}\right):\; a\in\F_q^*,b\in\F_q\right\}.
\end{equation}
The group $\AGL_2(q)$ is called the affine linear group. Every
element $A=\left(\begin{array}{cc}a&b\\ 0&1\end{array}\right)\in \AGL_2(q)$ defines an affine automorphism
$\Gs(x)=ax+b.$

\subsection{The subgroups of $\PGL_2(q)$}\label{subsec:2.3}
As we need subgroups of $\PGL_2(q)$ to construct optimal \LRCs in this paper, we have to find out subgroup structures of $\PGL_2(q)$. It is fortunate that
all subgroup structures of $\PGL_2(q)$ are known. In this subsection, we present these known results.

By an easy counting argument, we know that the order of $\GL_2(q)$ is $(q^2-1)(q^2-q)$. Thus, the order of $\PGL_2(q)$ is $q(q^2-1)$. The elements in $\PGL_2(q)$ with nonzero square determinants in $\F_q^*$ form a subgroup of $\PGL_2(q)$ which is called the projective special linear group and denoted by $\PSL_2(q)$.

The structure of subgroups of $\PSL_2(q)$ was well known and first found by Moore \cite{M04} and Wiman \cite{W99}. The reader may refer to Dickson's book \cite[Chapter XII]{D58} for the details. If char$(\F_q)=2$, then $\PGL_2(q)=\PSL_2(q)$. Hence, we have the following result (see \cite{LS09, SL10}).

\begin{prop} \label{even q}
All the nontrivial subgroups of $\PGL_2(q)$ for $q=2^s$ are listed as follows.
\begin{itemize}
\item [(i)] For each $d|(q\pm 1)$ with $d>2$, there is one conjugacy class of $q(q\mp 1)/2$ cyclic subgroups that are isomorphic to $\mathbb{Z}_d$ of order $d$.
\item [(ii)]  For each $d|(q\pm 1)$ with $d>2$,  there is one conjugacy class of $q(q^2-1)/(2d)$ subgroups that are isomorphic to the dihedral group $D_{2d}$ of order $2d$.
\item [(iii)]  For even $s$,  there is one conjugacy class of $q(q^2-1)/12$ subgroups that are isomorphic to the alternating group $A_4$.
\item [(iv)]  For even $s$,  there is one conjugacy class of $q(q^2-1)/60$ subgroups that are isomorphic to the alternating group $A_5$.
\item [(v)]  If $q$ is a power of $\ell$,  there is one conjugacy class of $\frac{q(q^2-1)}{\ell(\ell^2-1)}$ subgroups that are isomorphic to  $\PGL_2(\ell)$.
\item [(vi)] There are one conjugacy class of $q+1$ elementary abelian subgroups of order $q$, and $(q^2-1)/(p^k-1)$ conjugacy classes of elementary abelian subgroups of order $p^v$ for each positive integer $1\le v\le s-1$, where $k=\gcd(s,v)$.
\item [(vii)]  There are $(q^2-1)p^{s-v}/(p^k-1)$ conjugacy classes of subgroups of order $up^v$ which are semidirect products of an elementary abelian group of order $p^v$ and a cyclic group of order $u$, for every positive integer $1\le v\le s$ and every positive integer $u>2$ such that
$u|(p^k-1)$, where $k=\gcd(s,v)$.
\end{itemize}
Furthermore, the subgroups in (vi) and (vii) are contained in the subgroups that are isomorphic to $\AGL_2(q)$.
\end{prop}

Let $q=p^s$ be a power of an odd prime $p$ and let $\Ge$ be $1$ or $-1$ which is determined by $q\equiv \Ge$ (mod $4$).
The projective general linear group $\PGL_2(q)$ is a subgroup of $\PSL_2(q^2)$ and has a unique normal subgroup $\PSL_2(q)$ of index $2$.
Hence, the subgroup structures of $\PGL_2(q)$ can be determined (see \cite{COR06,LS09}).

\begin{prop} \label{odd q}
All the nontrivial subgroups of $\PGL_2(q)$ with $q=p^s$ for any odd prime $p$ are listed as follows.
\begin{itemize}
\item [(i)] There are two conjugacy classes of cyclic subgroups of order $2$.
\item [(ii)] There is one conjugacy class of $q(q\mp \Ge)/2$ cyclic subgroups  of size $d$ for every divisor $d>2$ of $(q\pm \Ge)$.
\item [(iii)] There are two conjugacy classes of $q(q^2-1)/6$  subgroups that are isomorphic to the dihedral group $D_4$ (or Klein $4$-group).
\item [(iv)]  There are two conjugacy classes of $q(q^2-1)/(2d)$ subgroups that are isomorphic to the dihedral group  $D_{2d}$ of order $2d$  for every divisor $d>2$ of $\frac{q\pm \Ge}{2}$.
\item [(v)]  There is one conjugacy class of  $q(q^2-1)/(2d)$ subgroups that are isomorphic to the dihedral group $D_{2d}$ of order $2d$  for every divisor $d>2$ of $(q\pm \Ge)$ such that $(q\pm \Ge)/d$ is an odd integer.
\item [(vi)] There are  $q(q^2-1)/24$ subgroups isomorphic to $A_4$; $q(q^2-1)/24$ subgroups isomorphic to $S_4$, and $q(q^2-1)/60$ subgroups isomorphic to $A_5$ when $q\equiv \pm 1(\text{mod } 10)$. 
\item [(vii)] There is one  conjugacy class of $\frac{q(q^2-1)}{\ell(\ell^2-1)}$ subgroups that are isomorphic to $\PSL_2(\ell)$,  where $q$ is a power of $\ell$.
\item [(viii)] $\PGL_2(q)$  contains the subgroups $\PGL_2(\ell)$, where $q$ is a power of $\ell$.
\item [(ix)]  $\PGL_2(q)$ contains the elementary abelian subgroups of order $p^v$ for $1\le v\le s$.
\item [(x)] $\PGL_2(q)$  contains the subgroups of order $up^v$ that are semidirect products of an elementary abelian subgroup of order $p^v$ with $1\le v\le s$ and a cyclic subgroup of order $u>1$ for every pair $(u,v)$ satisfying  $u|\gcd(q-1,p^v-1)$.
\end{itemize}
\end{prop}

\subsection{Hilbert's ramification theory}
For any subgroup $\mG$ of $\Aut(F/\F_q)$, let $F^\mG$ be the fixed subfield of $F$ with respect to $\mG$, i.e.,
\begin{equation}\label{eq:7}
F^\mG=\{y\in F:\; \Gs(y)=y\ \mbox{for all $\Gs\in\mG$}\}.
\end{equation}
By the Galois theory, the extension $F/F^\mG$ is a Galois extension with the Galois group $\Gal(F/F^\mG)=\mG$ (see \cite[Theorem 11.36]{HKT08}).

From the Hurwitz genus formula (see \cite[Theorem 3.4.13]{St09}), one has
$$2g(F)-2=|{\mathcal G}|\cdot [2g(F^{\mathcal G})-2]+\deg \text{ Diff}(F/F^{\mathcal G}),$$
where $g(F)$ stands for the genus of $F$ and  $\text{Diff}(F/F^{\mathcal G})$ is the different of $F/F^{\mathcal G}$.
Let $Q$ be a place of $F^{\mG}$ and let $P_1, \cdots, P_t$ be all the places of $F$ lying over $P$.
The Galois group $\mG$ acts transitively on the set $\{P_1, \cdots, P_t\}$ of extensions of  $Q$ in $F$ (see \cite[Theorem 3.7.1]{St09}).
Then the ramification indices $e_{P_i}(F/F^{\mG})$ (resp. different exponents $d_{P_i}(F/F^{\mG})$) of $P_i|Q$ are independent of $P_i$ for $1\le i\le t$, and we can denote them by $e_Q(F/F^{\mG})$ (resp. $d_Q(F/F^{\mG})$).
Hence, the degree of the different of $F/F^{\mathcal G}$ can be given by
$$\deg \text{Diff}(F/F^{\mathcal G})=\sum_{P\in  \PP_F} d_P(F/F^{\mG})\deg(P)=\sum_{Q\in\PP_{F^{\mG}}}  \frac{d_Q(F/F^{\mG})}{e_Q(F/F^{\mG})}\cdot |\mG| \cdot \deg(Q).$$

The $i$-th ramification group ${\mathcal G}_i(P)$ of $P|Q$ for each $i\ge -1$ is defined by
$${\mathcal G}_i(P)=\{ \sigma\in {\mathcal G}:\; v_P(\sigma(z)-z) \ge i+1 \text{ for all } z\in \mathcal{O}_P\},$$
where $\mathcal{O}_P$ is the integral ring of $P$ in $F$ and $v_P$ is the normalized discrete valuation of $F$ corresponding to $P$.
In fact, the order of $\mG_1(P)$ is a power of char($\F_q$), $\mG_0(P)/\mG_1(P)$ is a cyclic group whose order is relatively prime to char($\F_q$), $\mG_{i+1}(P)$ is a normal subgroup of $\mG_i$ for each $i\ge 1$,  and there exists an integer $a$ such that $\mG_i(P)=\{\mbox{id}\}$ for all $i\ge a$ (see \cite[Proposition 3.8.5]{St09}).
Moreover, the different exponent $d_Q(F/F^{\mG})$ is
$$d_Q(F/F^{\mG})=\sum_{i=0}^{\infty} \big{(} |{\mathcal G}_i(P)|-1\big{)}$$
from  Hilbert's Different Theorem \cite[Theorem 3.8.7]{St09}.  Furthermore, we have the following ramification structures for rational function fields (see \cite[Theorem 11.92]{HKT08}).

\begin{prop}\label{prop:2.3}
Let $F$ be the rational function field $\F_q(x)$ over $\F_q$.
Denote  by $\mA$ the automorphism group $\Aut(F/\F_q)$, then the extension $F/F^{\mA}$ has the following properties.
\begin{itemize}
\item [(i)] Every rational place $P$ of $F$ is wildly ramified in $F/F^{\mA}$; $\mG_{-1}(P)=\mG_0(P)$ is the semidirect product of an elementary abelian $p$-group of order $q$ with a cyclic group of order $q-1$; $\mG_1(P)$ is  an elementary abelian $p$-group of order $q$ and $\mG_i(P)$ is trivial for any $i\ge 2$. Furthermore, all rational places form a single orbit of size $q+1$ under the group action of  $\mA$ on the set of places of $F$.
\item [(ii)] Every place $Q$ of degree $2$ of $F$ is tamely ramified in $F/F^{\mA}$; $\mG_0(Q)$ is a cyclic group of order $q+1$. Furthermore, all places of degree $2$ form a single orbit of size $(q^2-q)/2$ under the group action of  $\mA$ on the set of places of $F$.
\item [(iii)]  Other places of $F$ are unramified in $F/F^{\mA}$.
\end{itemize}
\end{prop}

In fact, one can explicitly compute all the ramification groups of rational places and places of degree $2$. For instance, let $P_{\infty}$ be the pole of $x$. Then  $\mG_{-1}(P_{\infty})=\mG_0(P_{\infty})=\AGL_2(q)$ and $\mG_1(P_{\infty})$ is isomorphic to the group $\left\{\left(\begin{array}{cc}1&b\\ 0&1\end{array}\right):\; b\in\F_q\right\}\simeq (\F_q,+)$.

\subsection{Algebraic geometry codes} \label{AGcode}
In this subsection, we introduce a modification of the algebraic geometry codes. The reader may refer to \cite{NX01,TV91,TVN90} for details on algebraic geometry codes.

Let $F/\F_q$ be the rational function field. Let $\mP=\{P_1,\dots,P_n\}$ be a set of $n$ distinct rational places of $F$. For a divisor $G$ with $0<\deg(G)<n$ and $\supp(G)\cap\mP=\emptyset$, the algebraic geometry code is defined to be
\begin{equation}\label{eq:8}
C(\mP,G):=\{(f(P_1),\dots,f(P_n)): \; f\in\mL(G)\}.
\end{equation}
Then $C(\mP,G)$ is an $[n,\deg(G)+1,n-\deg(G)]$-MDS code.
If $V$ is a subspace of $\mL(G)$, we can define a subcode of $C(\mP,G)$ by
\begin{equation}\label{eq:9}
C(\mP,V):=\{(f(P_1),\dots,f(P_n)):\; f\in  V\}.
\end{equation}
The code $C(\mP,V)$ is usually no longer an MDS code, but the minimum distance is still lower bounded by $n-\deg(G)$.

For the purpose of our paper, we need to  modify the above construction. We can remove the condition that $\supp(G)\cap\mP=\emptyset$.
Assume that $G$ is a divisor of $F$ and let $m_i=\nu_{P_i}(G)$. Choose a local parameter $\pi_i$ of $P_i$ for each $i\in \{1,2,\cdots, n\}$.
Then for any nonzero function $f\in\mL(G)$, we have $\nu_{P_i}(\pi_i^{m_i}f)= m_i+\nu_{P_i}(f)\ge m_i-\nu_{P_i}(G)=0$.
Define a modified algebraic geometry code as follows
\begin{equation}\label{eq:10}
C(\mP,G):=\{((\pi_1^{m_1}f)(P_1),\dots,(\pi_n^{m_n}f)(P_n)):\; f\in\mL(G)\}.
\end{equation}
We claim that $C(\mP,G)$ is still an $[n,\deg(G)+1,n-\deg(G)]$-MDS code in this case.
It is sufficient to show that the Hamming weight of the codeword $$\big{(}(\pi_1^{m_1}f)(P_1),\dots,(\pi_n^{m_n}f)(P_n)\big{)}$$ is at least $n-\deg(G)$ for every nonzero function $f\in\mL(G)$.
Let $I$ be the subset of $\{1,2,\dots,n\}$ such that $(\pi_i^{m_i}f)(P_i)=0$.
Then we have $f\in\mL(G-\sum_{i\in I}P_i)$. Thus, $\deg(G-\sum_{i\in I}P_i)\ge 0$, i.e., $|I|\le \deg(G)$. This gives a lower bound  $n-|I|\ge n-\deg(G)$ on the Hamming weight of the codeword.
Now, for a subspace $V$ of $\mL(G)$, we can define a subcode of $C(\mP,G)$ by
\begin{equation}\label{eq:11}
C(\mP,V):=\{((\pi_1^{m_1}f)(P_1),\dots,(\pi_n^{m_n}f)(P_n)):\; f\in V\}.
\end{equation}
Again, the minimum distance of $C(\mP,V)$ is  lower bounded by $n-\deg(G)$.

\section{General construction of optimal LRC codes}\label{sec:3}
The idea of our construction works as follows. Let $F/\F_q$ be the rational function field.
Let $\mG$ be a subgroup of $\Aut(F/\F_q)$ of order $r+1$. Then there is a subfield $E$ of $F$ such that $F/E$ is a Galois extension with the Galois group $\Gal(F/E)=\mG$. Now assume that $Q_1,Q_2,\dots, Q_m$ are rational places of $E$ and they are all splitting completely in $F/E$. Let $P_{i,1},P_{i,2},\dots,P_{i,r+1}$ denote the $r+1$ rational places of $F$ that lie over $Q_i$ for all $1\le i\le m$. Put $n=(r+1)m$ and put $\mP=\{P_{ij}\}_{1\le i\le m,1\le j\le r+1}$. Choose a divisor $G$ of degree $k+t-2$ with $k=rt$ for some integer $1\le t\le m$, a function $z$ of $E$ with $\deg(z)_{\infty}=1$ (here $(z)_{\infty}$ is considered as a divisor of $E$) and a function  $x$  of $F$ with $\deg(x)_{\infty}=1$. Then the code $\{(f(P))_{P\in\mP}:\; f\in \sum_{i=0}^{r-1}\left(\sum_{j=0}^{t-1}a_{ij}z^j\right)x^i;\; a_{ij}\in\F_q\}$ is an optimal $q$-ary $[n,k,d]$-\LRC with locality $r$,  $k=rt$ and $d=n-k-\frac kr+2$ (see the proof of Theorem \ref{prop:3.1} below). Apparently, to construct optimal locally repairable codes using this idea, one has to analyze the splitting behavior of places of $E$ in $F$. In order to do so, we need to have explicit structure of the subgroup $\mG$. As all subgroup structures are given in  Propositions \ref{even q} and \ref{odd q}, we are able to produce  optimal locally repairable codes from an arbitrary subgroup of $\PGL_2(q)$. On the other hand, it is unnecessary to provide explicit constructions from all subgroups of $\PGL_2(q)$ one by one. Instead, in this section, we give a general construction of optimal locally repairable codes for all the subgroups of automorphism groups of rational function fields over finite fields by estimating the number of ramified rational places in $F$. Explicit constructions corresponding to subgroups of the affine linear group $\AGL_2(q)$, the cyclic group of size $q+1$ and the dihedral groups $D_{2(q\pm 1)}$ for even $q$  are provided in the following sections.

\begin{prop}\label{prop:3.1} Assume that there is a subgroup of $\PGL_2(q)$ of order $r+1$ with $1\le r\le \frac{q}{3}$.
 Put $n=m(r+1)$ for any positive integer $m\le \left\lfloor\frac{q+1-2r}{r+1}\right\rfloor$. Then, for any integer $t$ with  $1\le t\le m$, there exists an optimal $q$-ary $[n,k,d]$-\LRC with locality $r$,  $k=rt$ and $d=n-k-\frac kr+2$.
\end{prop}
\begin{proof} Let $F/\F_q$ be the rational function field. Then there exists a subgroup  $\mG$  of $\Aut(F/\F_q)$ of order $r+1$.
Denote by $\{R_1,R_2,\dots,R_s\}$ the set of rational places of $F$ which are ramified in $F/F^\mG$.
By the Hurwitz genus formula, we have
\[2g(F)-2\ge (r+1)(2g(F^\mG)-2)+\sum_{i=1}^s d_{R_i}(F/F^{\mG})\deg(R_i)\ge -2(r+1)+s.\]
This gives $s\le 2r$.
Hence, there are at least $q+1-s\ge q+1-2r\ge m(r+1)$ rational places of $F$ which are unramified in $F/F^{\mG}$.
As a rational place of $F$ is either ramified or splitting completely in $F/F^{\mG}$, there are $m$ rational places $\{Q_{1}, Q_{2}, \cdots, Q_{m} \}$ of $F^{\mG}$ which split completely in $F/F^{\mG}$. Denote by $\mP=\{P_1,P_2,\dots,P_n\}$ the set of rational places of $F$ lying over $\{Q_{1}, Q_{2}, \cdots, Q_{m} \}$.

Let $P_{\infty}$ be a rational place of $F$ such that $P_{\infty}\notin \mP$. Then there exists an element $x\in F$ such that $(x)_{\infty}=P_{\infty}$ and $F=\F_q(x)$ (see \cite[Proposition 1.6.6]{St09}).
Let $Q_{\infty}$ be the restriction of $P_{\infty}$ to $F^{\mG}$.
Similarly, there exists an element $z\in F^{\mG}$ such that $(z)_{\infty}=Q_{\infty}$ in $F^{\mG}$ and $F^{\mG}=\F_q(z)$. Thus, $\supp((z)_{\infty})\cap\mP=\emptyset$. Note that as a divisor of $F^{\mG}$, we have $(z)_{\infty}=Q_{\infty}$. However, as a divisor of $F$, we have $\deg(z)_{\infty}=[F:F^{\mG}]=r+1$.

For $t\ge 1$, consider the set of functions
\[V:=\left\{\sum_{i=0}^{r-1}\left(\sum_{j=0}^{t-1}a_{ij}z^j\right)x^i\in F:\; a_{ij}\in\F_q\right\}.\]
First of all, as $[\F_q(x):\F_q(z)]=r+1$, the set $\{1,x,\dots,x^{r-1}\}$ is linearly independent over $\F_q(z)$. Thus, the vector space $V$ over $\F_q$ has dimension $rt$.
Consider the subcode of the algebraic geometry  code
\[C(\mP,V)=\{(f(P_1),\dots,f(P_n)):\; f\in V\}.\]
Then we claim that $C(\mP,V)$ is an optimal $[n,k,d]$-\LRC with locality $r$, $k=rt$  and $d=n-k-\frac kr+2$.

Firstly, it is clear that $C(\mP,V)$ is an $\F_q$-linear code of length $n$. For every nonzero function $f\in V$,   the pole  divisor of $f$ is at most $(t-1)(z)_{\infty}+(r-1)(x)_{\infty}$. Thus, $\deg(f)_{\infty}\le (t-1)(r+1)+(r-1)=(r+1)t-2$ (note that we consider divisors of $F$ here, and hence $\deg(z)_{\infty}=[F:F^{\mG}]=r+1$). Thus, $f$ has at most $(r+1)t-2$ zeros. This implies that the Hamming weight of $(f(P_1),\dots,f(P_n))$ is at least $n-rt-t+2=n-k-\frac kr+2>0$. Therefore, the dimension of $C(\mP,V)$ is $k=rt$ and $d\ge n-k-\frac kr+2$. By the Singleton-type bound \eqref{eq:x1}, we obtain $d=n-k-\frac kr+2$ if $C(\mP,V)$ has locality $r$.

Finally let us prove that $C(\mP,V)$ has locality $r$. We claim that $x(P_\Ga)\neq x(P_\Gb)$ for any two different rational places $P_\Ga,P_\Gb\in \mP$. Otherwise,  $x-c$ would have two zeros $P_\Ga$ and $P_\Gb$, where $c=x(P_\Ga)= x(P_\Gb)\in \F_q$. Thus, $\deg(x-c)_0\ge 2$. This is a contradiction to the fact that $\deg(x-c)_0=\deg(x-c)_\infty=\deg(x)_\infty=1$.

Let $A_\ell$ be the set of rational places of $F$ which lie over $Q_\ell$ for $1\le \ell\le m$.
Suppose that the erased symbol of the codeword is $c_{\Ga}=f(P_{\Ga})$, where $P_{\Ga}\in A_\ell$. Put $b_i=\sum_{j=0}^{t-1}a_{ij}z^j(P_{\Ga})$ for all $0\le i\le t-1$. As $z$ is a function of $F^\mG$, we have $z(P_\Gb)=z(P_{\Ga})$ for all $\Gb\in A_\ell$. Hence, for any $P_{\Gb}\in A_\ell$,
we have $$f(P_{\Gb})=\sum_{i=0}^{r-1}\left(\sum_{j=0}^{t-1}a_{ij}z^j(P_{\Gb})\right)x^i(P_{\Gb})=\sum_{i=0}^{r-1}b_ix^i(P_\Gb).$$
Define the decoding polynomial $\Gd(x)=\sum_{i=0}^{r-1}b_ix^i.$
Then $$\Gd(x)(P_{\Gb})=\sum_{i=0}^{r-1}b_ix^i(P_{\Gb})=f(P_{\Gb}).$$
Since $\Gd(x)$ is a polynomial of degree at most $r-1$ and $\{x(P_{\Gb})\}_{\Gb\in A_\ell\setminus \{P_{\Ga}\}}$ are pairwise distinct, it can be interpolated from these $r$ symbols $c_{\Gb}=f(P_{\Gb})$ for $P_{\Gb}\in A_{\ell}\setminus \{P_{\Ga}\}$. Hence, the erased symbol $c_{\Ga}=\Gd(x)(P_{\Ga})$  can be recovered by the Lagrange interpolation. This completes the proof.
\end{proof}

Combining Proposition \ref{prop:3.1} with Propositions \ref{even q} and \ref{odd q}, we immediately obtain the following optimal locally repairable codes.

\begin{theorem}\label{thm:3.2}
If a positive integer $r\le \frac{q}{3}$ satisfies one of the following conditions, then for any positive integer $n\le (r+1)\left\lfloor\frac{q+1-2r}{r+1}\right\rfloor$ that is divisible by $r+1$ and any integer $t$ with  $1\le t\le \frac{n}{r+1}$, there exists an optimal $q$-ary $[n,k,d]$-\LRC with locality $r$,  $k=rt$ and $d=n-k-\frac kr+2$.
\begin{itemize}
\item[{\rm (i)}] $r+1$ is a divisor of $q-1$;
\item[{\rm (ii)}] $r+1$ is a divisor of $q$;
\item[{\rm (iii)}] $r+1=up^v$, where $u\ge 2$, $v\ge 1$ and $u$ is a common divisor of $q-1$ and $p^v-1$;
\item[{\rm (iv)}] $r+1$ is a divisor of $q+1$;
\item[{\rm (v)}] $r+1=2h$, where $h\ge 2$ is a divisor of $q\pm 1$ for even $q$ or $\frac{q\pm 1}{2}$ for odd $q$;
\item[{\rm (vi)}] $r+1=\ell(\ell^2-1)$, where $q$ is a power of $\ell$;
\item[{\rm (vii)}] $r+1=\ell(\ell^2-1)/2$, where $q$ is odd and $q$ is a power of $\ell$;
\item[{\rm (viii)}] $r+1=12$, if $q$ is an even power of $2$ or $q$ is odd.
\item[{\rm (ix)}] $r+1=24$, if $q$ is odd;
\item[{\rm (x)}] $r+1=60$, if $q$ is an even power of $2$ or $q\equiv \pm 1(\text{mod } 10)$.
\end{itemize}
\end{theorem}

\begin{rmk}{\rm The advantage of  Theorem \ref{thm:3.2} is that one can produce optimal $q$-ary \LRCs  without knowing explicit structures of subgroups of $\PGL_2(q)$, while the disadvantage of Theorem \ref{thm:3.2} is that it is an existence result.
}\end{rmk}

\begin{rmk}
For a subgroup $\mG$ of $\Aut(F/\F_q)$ of order $r+1$, let $s\ge 1$ be the number of rational places of $F$ which are ramified in $F/F^{\mG}$ (see \cite[Theorem 11.91]{HKT08}). Then, the length $n$ of the optimal \LRCs can be as large as $q+1-s$. In the following sections, by analyzing explicit subgroup structures, we will construct the optimal locally repairable codes explicitly.
\end{rmk}

Let us look at some examples by plugging some values of $q$, $r$, $t$ and $n$ in Theorem \ref{thm:3.2}.

\begin{ex}\label{ex:3.3}
{\rm \begin{itemize}
\item[{\rm (i)}]
Let $q=27$. Then for $r\in\{1,2,3,5,6,8\}$, there is a subgroup of $\PGL_2(27)$ of order $r+1$. Thus, for  $n= (r+1)\left\lfloor\frac{28-2r}{r+1}\right\rfloor$  and
any integer $t$ with  $1\le t\le \left\lfloor\frac{28-2r}{r+1}\right\rfloor$, there exists an optimal $q$-ary $[n,k,d]$-\LRC with locality $r$,  $k=rt$ and $d=n-k-\frac kr+2$. For instance, for $r=2$, there exists an optimal $27$-ary $[24,2t,26-3t]$-\LRC with locality $3$ for any $1\le t\le 8$.
If we take $r=3$, we obtain  an optimal $27$-ary $[20,3t,22-4t]$-\LRC with locality $3$ for any $1\le t\le 5$.

\item[{\rm (ii)}]  Let $q=64$. Then for $r\in\{1,2,3,4,5,6,7,8,9,11,12,13,15,17,20\}$, there is a subgroup of $\PGL_2(64)$ of order $r+1$. Thus, for  $n= (r+1)\left\lfloor\frac{65-2r}{r+1}\right\rfloor$  and
any integer $t$ with  $1\le t\le \left\lfloor\frac{65-2r}{r+1}\right\rfloor$, there exists an optimal $q$-ary $[n,k,d]$-\LRC with locality $r$,  $k=rt$ and $d=n-k-\frac kr+2$. For instance, for $r=3$, there exists an optimal $64$-ary $[56,3t,58-4t]$-\LRC with locality $3$ for any $1\le t\le 14$. If we let $r=5$, then there exists an optimal $64$-ary $[54,5t,56-6t]$-\LRC with locality $5$ for any $1\le t\le 9$.
\end{itemize}
}\end{ex}

By Proposition \ref{prop:3.1}, Propositions \ref{even q} and \ref{odd q}, we can give a result on some small locality for \LRCs over finite fields of small characteristics.

\begin{cor}
If a pair $(q,r)$ of positive integers satisfies one of the following conditions, then for any positive integer $n\le (r+1)\left\lfloor\frac{q+1-2r}{r+1}\right\rfloor$ that is divisible by $r+1$ and
 any integer $t$ with  $1\le t\le \frac{n}{r+1}$, there exists an optimal $q$-ary $[n,k,d]$-\LRC with locality $r$,  $k=rt$ and $d=n-k-\frac kr+2$.
\begin{itemize}
\item[{\rm (i)}] $q=2^s$ with $s\ge 4$ and $r\in\{1,2,5\}$;
\item[{\rm (ii)}] $q=3^s$ with $s\ge 4$ and $r\in\{1,2,3,5,11,23\}$;
\item[{\rm (iii)}] $q=4^s$ with $s\ge 3$ and $r\in\{1,2,3,4,5,9,11\}$;
\item[{\rm (iv)}] $q=5^s$ with $s\ge 3$  and $r\in\{1,2,3,4,5,9,11,19,23\}$.
\end{itemize}
\end{cor}
\begin{proof} (i) follows from Proposition \ref{prop:3.1} and the fact that $\PGL_2(q)$ contains the group $\PGL_2(2)$ and $\PGL_2(2)$ contains subgroups of order $2$, $3$ and $6$.

(ii) For $q=3^s$ with $s\ge 4$, $2|(q-1)$, $3|q$, $4|(q-1)$ for even $s$ or $4|(q+1)$ for odd $s$. Thus, there is a subgroup of $\PGL_2(q)$ of order $r+1$ for $r\in\{1,2,3\}$. There is also an affine subgroup of order $6$ since $2|(3-1)$ from Theorem \ref{thm:3.2}(iii).  The order of $A_4$ is $12$, and the order of $\PGL_2(3)$ is $24$.

(iii) There is a subgroup of  $\PGL_2(q)$ of order $r+1$ for $r\in\{1,2,3,4,5,9,11\}$ since  $2|q$, $3|(q-1)$, $4|q$, $5|(q-1)$ for even $s$ or $5|(q+1)$ for odd $s$; the order of $\PGL_2(2)$ is 6; there is a dihedral subgroup of order $10$ from  Theorem \ref{thm:3.2}(v); and the order of $A_4$ is $12$.

(iv) There is a subgroup of  $\PGL_2(q)$ of order $r+1$ for $r\in\{1,2,3,4,5,9,11,19,23\}$ since  both $2$ and $4$ are divisors of $q-1$; both $3$ and $6$ are divisors of $q-1$ for even $s$ or $q+1$ for odd $s$; $5|q$; there are affine subgroups of order $10$ and $20$, respectively, since $2|(5-1)$ and $4|(5-1)$ from  Theorem \ref{thm:3.2}(iii); the order of $A_4$ is $12$; and the order of $S_4$ is $24$.
\end{proof}

\section{Explicit construction via affine subgroups}\label{sec:4}
In the previous section, we provided a construction for arbitrary subgroups of $\PGL_2(q)$ by estimating the number of ramified rational places. From this section onwards, we will consider some particular subgroups and analyze ramification of rational places.

In this section, we will construct optimal \LRCs from subgroups of $\AGL_2(q)$. It turns out that the optimal \LRCs constructed in  \cite{TB14} are examples in this section.

Let $\F_q$ be the finite field with $q=p^s$ elements and let $F$ be the rational function field $\F_q(x)$ over $\F_q$. The proof of the following proposition provides explicit group structures of affine linear group  $\AGL_2(q)$.

\begin{prop} Let $q=p^s$.
Let $v$ be an integer with $0\le v\le s$ and let $u$ be a positive integer satisfying $u|(q-1)$ and $u|(p^v-1)$.
Then there is a subgroup $\mG$ of  $\AGL_2(q)$  of order $up^v$.
\end{prop}

\begin{proof}
If $u$ is a divisor of $q-1$, then there exists a subgroup $H$ of the multiplicative group $\F_q^*$ of order $u$. As $u|(p^v-1)$,  the field $\F_p(H)$ is contained in $\F_{p^v}$. Put $\ell=\min \{ t>0: u|(p^t-1)\}$. Then we have
 $\F_p(H)=\F_{p^\ell}$ and $\ell|\gcd(v,s)$.

Let $W$ be a vector subspace of $\F_q$ over $\F_{p^\ell}$ with dimension $v/\ell$. Put
\begin{equation}\label{up^u} \mG:=\left\{\left(\begin{array}{cc}a&b\\ 0&1\end{array}\right):\; a\in H,\ b\in W\right\}.\end{equation}
Then it is easy to verify that $\mG$ is a subgroup of $\AGL_2(q)$ of order $up^{v}$.
\end{proof}

The ramification information of the $F/F^{\mG}$ are provided in the following proposition. In particular, the number of ramified rational places of $F$ in the extension  $F/F^{\mG}$ can be determined.

\begin{prop}\label{prop:4.2} Let $q=p^s$.
Let $\mG$ be a subgroup of $\AGL_2(q)$ with order $up^{v}$ that is defined in Equation \eqref{up^u}. Then the extension $F/F^{\mG}$ has the following properties.
\begin{itemize}
\item [(i)] $[F:F^{\mG}]=up^{v}$.
\item [(ii)] The infinity place $P_{\infty}$ of $F$ is totally ramified in $F/F^{\mG}$ with ramification index $e_{P_{\infty}}(F/F^{\mG})=up^{v}$ and different exponent  $d_{P_{\infty}}(F/F^{\mG})=up^{v}+p^{v}-2$, where $P_{\infty}$ is the unique pole place of $x$. There is a rational place of $F^{\mG}$ which splits into $p^{v}$ rational places $\{P_1, P_2, \cdots, P_{p^{v}}\}$ of  $F$, each place $P_i$ has ramification index $e_{P_{i}}(F/F^{\mG})=u$ and different exponent $d_{P_{i}}(F/F^{\mG})=u-1$ in $F/F^{\mG}$. Furthermore, there are $(q-p^v)/(up^v)$ rational places of $F^{\mG}$ that split completely in $F/F^\mG$.
\item [(iii)] The unique zero place of $x$ is ramified  in  $F/F^{\mG}$ with ramification index $u$. Hence, $P_0\in\{P_1, P_2, \cdots, P_{p^{v}}\}$.
\item[(iv)] Other places of $F$ are unramified in  $F/F^{\mG}$.
\end{itemize}
\end{prop}
\begin{proof} (i) is clear.

To prove (ii), let $\mA$ denote the automorphism group $\Aut(F/\F_q)$.
From Proposition \ref{prop:2.3} and the paragraph after Proposition \ref{prop:2.3}, we know that the inertia group of  the infinite place  $P_{\infty}$ in $F/F^\mA$ is ${\AGL_2(q)}$. Thus,  the infinite place $P_{\infty}$ is totally ramified in $F/F^{\AGL_2(q)}$.
Since $\mG$ is a subgroup of $\AGL_2(q)$, $P_{\infty}$ is totally ramified in $F/F^{\mG}$, i.e., $e_{P_{\infty}}(F/F^{\mG})=up^{v}$.
It is straightforward to verify that the orders of  ramification groups of $P_{\infty}$ in $F/F^{\mG}$ are given by
$$|\mG_0(P_{\infty})|= up^v,  |\mG_1(P_{\infty})|= p^v \text{ and }  |\mG_i(P_{\infty})|= 1 \text{ for } i\ge 2.$$
Hence, the different exponent is $d_{P_{\infty}}(F/F^{\mG})=up^{v}+p^{v}-2$ from Hilbert's Different Theorem.
Assume that $Q_1,\cdots, Q_k$ are the remaining ramified places of $F^{\mG}$ in $F/F^{\mG}$ with ramification indices $e_i$ and different exponents $d_i$ for $1\le i\le k$.
Form the Hurwitz genus formula, one has
\begin{equation}\label{eq:x4}2g(F)-2=|\mG| \cdot [2g(F^{\mG})-2]+up^{v}+p^{v}-2+\sum_{i=1}^{k} \frac{d_i}{e_i}\cdot |\mG|.\end{equation}
Since any fixed subfield of the rational function field is again a rational function field, we have $g(F)=g(F^\mG)=0$. Hence, it follows from \eqref{eq:x4} that $\sum_{i=1}^{k} \frac{d_i}{e_i}=\frac{u-1}{u}$. If $u=1$, then $k=0$. Otherwise we must have $k=1$ and $e_1=u$, and the place $Q_1$ splits into $p^v$ rational places in $F$.
Then there are $q-p^v$ rational places of $F$ which are unramified in $F/F^{\mG}$.
Hence, there are $(q-p^v)/(up^v)$ rational places of $F^{\mG}$ that split completely in $F/F^\mG$.

Again, it is straightforward to verify that the orders of  ramification groups of $P_{0}$ in $F/F^{\mG}$ are given by
$|\mG_0(P_{0})|= u,  \text{ and }  |\mG_i(P_{0})|= 1 \text{ for } i\ge 1.$
Thus, $P_0$ is ramified  in $F/F^\mG$ with ramification index $u$. This proves (iii).

(iv) follows from the proof of (ii).
\end{proof}

Now we can provide  explicit constructions of optimal \LRCs from the subgroups of $\AGL_2(q).$

\begin{theorem}\label{thm:4.3} Let $q=p^s$.  Let $v$ be a positive integer less than or equal to $s$.
Put $r=p^v-1$ and $n=m(r+1)$ for any positive integer $m\le q/(r+1)$.
Then for any integer $t$ with $1\le t \le m$, there exists an optimal $[n,k,d]_q$-locally repairable code with locality $r$, $k=rt$ and $d=n-k-\frac{k}{r}+2$.
\end{theorem}
\begin{proof}
Let $F$ be the rational function field $\F_q(x)$ and let $P_{\infty}$ be the infinite place of $F$, i.e., $(x)_{\infty}=P_{\infty}$.
Let $\mG$ be the subgroup of $\AGL_2(q)$ with order $r+1=p^{v}$ constructed from Equation (\ref{up^u}). Then by Proposition \ref{prop:4.2}, except for $P_{\infty}$ all other $q$ rational places of $F$ split completely.
Denote by $\mP$ a set of  $m(r+1)$ rational places of $F$ that lie over $m$ rational places of $F^\mG$.

Let $Q_{\infty}$ be the place of $F^\mG$ that lies under $P_{\infty}$ and choose an element $z\in F^\mG$ such that $(z)_{\infty}$ is equal to $Q_{\infty}$ as a divisor of $F^\mG$. Thus, $(z)_{\infty}$ is equal to $(r+1)P_{\infty}$ as a divisor of $F$.

For $t\ge 1$, consider the set of functions
\[V:=\left\{\sum_{i=0}^{r-1}\left(\sum_{j=0}^{t-1}a_{ij}z^j\right)x^i:\; a_{ij}\in\F_q\right\}.\]
Then by mimicking  the proof of Proposition \ref{prop:3.1}, one can show that the code $C(\mP, V)$ is an optimal $[n,k,d]_q$-locally repairable code with locality $r$, $k=rt$ and $d=n-k-\frac{k}{r}+2$.
\end{proof}

\begin{rmk} In the proof of Theorem \ref{thm:4.3}, one can choose the function $z$ to be $\prod_{b\in W}(x+b)$, where $W$ is the $\F_p$-vector space of dimension $v$ used to define $\mG$ in Equation \eqref{up^u}. Thus, Theorem \ref{thm:4.3} provides the same construction as the one using additive subgroups in \cite[Proposition 3.2]{TB14}.
\end{rmk}

In Theorem \ref{thm:4.3}, we make use of a subgroup $\mG$ of $\AGL_2(q)$ that is isomorphic to an additive subgroup of $\F_q$. The following theorem will use a subgroup $\mG$ of $\AGL_2(q)$ that is isomorphic to a multiplicative  subgroup of $\F_q^*$.

\begin{theorem}\label{thm:4.4} Let $r$ be a positive integer with $(r+1)|(q-1)$.
Put $n=m(r+1)$ for any positive integer $m\le (q-1)/(r+1)$.
Then for any integer $t$ with $1\le t \le m$, there exists an optimal $[n,k,d]_q$-locally repairable code with locality $r$, $k=rt$ and $d=n-k-\frac{k}{r}+2$.
\end{theorem}
\begin{proof}
Let $F$ be the rational function field $\F_q(x)$ and let $P_{\infty}$ be the pole place of $x$, i.e., $(x)_{\infty}=P_{\infty}$.
Let $\mG$ be the subgroup of $\AGL_2(q)$ with order $r+1$ constructed from Equation (\ref{up^u}). Then by Proposition \ref{prop:4.2}, except for $P_{\infty}$ and the zero of $x$, all other $q-1$ rational places of $F$ split completely.
Denote by $\mP$ a set of  $m(r+1)$ rational places of $F$ that lie over $m$ rational places of $F^\mG$.

Let $Q_{\infty}$ be the place of $F^\mG$ that lies under $P_{\infty}$ and choose an element $z\in F^\mG$ such that $(z)_{\infty}$ is equal to $Q_{\infty}$ as a divisor of $F^\mG$. Thus, $(z)_{\infty}$ is equal to $(r+1)P_{\infty}$ as a divisor of $F$.

For $t\ge 1$, consider the set of functions
\[V:=\left\{\sum_{i=0}^{r-1}\left(\sum_{j=0}^{t-1}a_{ij}z^j\right)x^i:\; a_{ij}\in\F_q\right\}.\]
Then by mimicking  the proof of Proposition \ref{prop:3.1}, one can show that the code $C(\mP, V)$ is an optimal $[n,k,d]_q$-locally repairable code with locality $r$, $k=rt$ and $d=n-k-\frac{k}{r}+2$.
\end{proof}

\begin{rmk} In the proof of Theorem \ref{thm:4.3}, one can choose the function $z$ to be $\prod_{a\in H}ax=x^{r+1}$, where $H$ is the subgroup of $\F_q^*$ used to define $\mG$ in Equation \eqref{up^u}. Thus, Theorem \ref{thm:4.4} provides the same construction as the one using multiplicative subgroups in \cite[Proposition 3.2]{TB14}.
\end{rmk}

Now we consider subgroups of $\AGL_2(q)$ that mix additive subgroups of $\F_q$ and multiplicative subgroups of $\F_q^*$.

\begin{theorem}\label{thm:4.5}
Let $u\ge 2$ be a common divisor of $q-1$ and $p^v-1$ for some $1\le v\le s$. Let $r$ be a positive integer such that $r+1=up^v$.
Put $n=m(r+1)$ for  any positive integer $m\le (q-p^v)/(up^v)$.
Then for any integer $t$ with $1\le t \le m$, there exists an optimal $[n,k,d]_q$-locally repairable code with locality $r$, $k=rt$ and $d=n-k-\frac{k}{r}+2$.
\end{theorem}
One can imitate the proof of Theorem \ref{thm:4.3} or Theorem \ref{thm:4.4} by considering a subgroup of $\AGL_2(q)$ of size $up^v$ defined in Equation (\ref{up^u}). We skip the details.


\begin{ex} {\rm
 \begin{itemize}
 \item[{\rm (i)}]
Let $q=27$. If we let $r=2$ and $m=9$ in  Theorem \ref{thm:4.3}, then we can explicitly construct an optimal $27$-ary $[27,2t,29-3t]$-\LRC with locality $2$ for any $1\le t\le 9$. Note that the code length gets enlarged compared with Example \ref{ex:3.3}(i). This is because we know the ramification situation.
If we let $r=12$ and $m=2$ in  Theorem \ref{thm:4.4}, then we can explicitly construct an optimal $27$-ary $[26,12t,28-13t]$-\LRC with locality $12$ for any $1\le t\le 2$.
If we let $r=5$ and $m=4$ in  Theorem \ref{thm:4.5}, then we can explicitly construct an optimal $27$-ary $[24,5t,26-6t]$-\LRC with locality $5$ for any $1\le t\le 4$.
\item[{\rm (ii)}] Let $q=64$. If we let $r=3$ and $m=16$ in  Theorem \ref{thm:4.3}, then we can explicitly construct an optimal $64$-ary $[64,3t,66-4t]$-\LRC with locality $3$ for any $1\le t\le 16$. Note that the code length gets enlarged compared with Example \ref{ex:3.3}(ii).
If we let $r=8$ and $m=7$ in  Theorem \ref{thm:4.4}, then we can explicitly construct an optimal $64$-ary $[63,8t,65-9t]$-\LRC  with locality $8$ for any $1\le t\le 7$.
If we let $r=11$ and $m=5$ in  Theorem \ref{thm:4.5}, then  an optimal $64$-ary $[60,11t,62-12t]$-\LRC  with locality $11$ for any $1\le t\le 5$ can be constructed.
\end{itemize}
}\end{ex}

\section{Explicit construction of LRC codes of length $n=q+1$}
In this section, we give an explicit construction of optimal \LRCs of length $n=q+1$ via the subgroups of a cyclic group of order $q+1$ of $\Aut(F/\F_q)$. First of all, we provide an explicit characterization of the fixed subfields of the rational function field $\F_q(x)$ with respect to the subgroups of a cyclic group of order $q+1$.

\begin{lemma}\label{subgroupofq+1}
Let $F$ be the rational function field $\F_q(x)$.  Let $f(x)$ be a quadratic primitive polynomial $x^2+ax+b\in \F_q[x]$ of order $q^2-1$. Then its companion matrix $A_f=\left(\begin{array}{cc}0 & 1 \\-b & -a\end{array}\right)$ generates a cyclic group of order $q+1$  in $\PGL_2(q)$. Furthermore, let
 $\mG$ be a subgroup of $\langle A_f\rangle$. Then we have
$F^{\mG}=\F_q(\mbox{\Tr}(x)),$  where $\mbox{\Tr}(x)=\sum_{\Gt\in \mG} \Gt(x)$.
\end{lemma}
\begin{proof}

Let us first prove that  $A_f$ has order $q+1$ in  $\PGL_2(q)$.

Let $\Ga$ and $\Ga^q$ be the two distinct roots of $f(x)$ in $\F_{q^2}$. Then $\text{ord}(\Ga)=q^2-1$, $\Ga+\Ga^q=-a$ and $\Ga^{q+1}=b$. Put $P=\left(\begin{array}{cc}1 & 1 \\ \Ga & \Ga^q\end{array}\right)$. Then we have
$P^{-1}A_fP=\left(\begin{array}{cc}\Ga & 0 \\0 & \Ga^q\end{array}\right).$
It is easy to verify that $q+1$ is the smallest positive integer $k$ such that $\Ga^k=\Ga^{qk}=c$ for some $c\in \F_q^*$.
Thus, $\left(\begin{array}{cc}\Ga & 0 \\0 & \Ga^q\end{array}\right)$ has order $q+1$ and hence $A_f$ generates a cyclic group of order $q+1$ in $\PGL_2(q)$.

Let $a_k=(\Ga^{qk}-\Ga^k)/(\Ga^q-\Ga)$. Then we obtain
\begin{eqnarray*}
A_f^k &=&P\cdot \left(\begin{array}{cc}\Ga^k & 0 \\0 & \Ga^{qk}\end{array}\right) \cdot P^{-1} \\
&=& \frac{1}{\Ga^q-\Ga}\left(\begin{array}{cc}\Ga^{q+k}-\Ga^{qk+1} & \Ga^{qk}-\Ga^k \\ \Ga^{q+k+1}-\Ga^{q(k+1)+1}  & \Ga^{q(k+1)}-\Ga^{k+1} \end{array}\right)\\
&=& \left(\begin{array}{cc}-ba_{k-1} & a_k \\-ba_k & a_{k+1}\end{array}\right).
\end{eqnarray*}
Moreover, it is easy to see that $a_{k}$ satisfies the following recursive formula
$$a_{k+2}+aa_{k+1}+ba_k =0 \text{ with } a_0=0, a_1=1.$$
Let $\Gs$ be the automorphism of $\F_q(x)$ corresponding to the matrix $A_f$, i.e.,
$\Gs(x)=1/(-b x-a)$. Then the order of $\Gs$ is $q+1$.
For $1\le k \le q$, we have
\begin{equation}\label{sigma_k}
\Gs^k(x)=\frac{-ba_{k-1}x+a_k}{-b a_k x+a_{k+1}}.
\end{equation}
It is easy to verify that $a_k=0$ if and only if $\Ga^{k}=\Ga^{qk}$, i.e., $(q+1)|k$. Hence, $a_k\neq 0$ for all $1\le k\le q$.
By a direct computation, one has  $$\frac{a_{k+1}}{ba_k}\neq \frac{a_{j+1}}{ba_j}$$ for all $1\le k\neq j \le q$.

Let $\mG$ be a subgroup of $\langle \Gs \rangle$.
Then $\mbox{Tr}(x)=\sum_{\Gt\in \mG} \Gt(x)\in F^{\mG}$ and $\deg(\mbox{Tr}(x))_{\infty}=|\mG|$.
By  \cite[Theorem 1.4.10]{St09}, we have $$|\mG|=[F:F^{\mG}]\le [F:\F_q(\mbox{Tr}(x))]=\deg(\mbox{Tr}(x))_{\infty}=|\mG|.$$
Hence, $F^{\mG}=\F_q(\mbox{Tr}(x)).$
\end{proof}

The ramification properties of $F/F^{\mG}$ are given in the following proposition and the number of the ramified rational places of $F$ in $F/F^{\mG}$ can be completely determined.

\begin{prop}\label{prop:5.2}
Let $\Gs$ be an automorphism of $F$ of corresponding to the matrix $A_f$ defined in Lemma \ref{subgroupofq+1}.
Let $\mG$ be a subgroup of $\langle \Gs \rangle$ of order $r+1$ such that $(r+1)|(q+1)$. Then the extension $F/F^{\mG}$ has the following properties.
\begin{itemize}
\item [(i)] $[F:F^{\mG}]=r+1$.
\item [(ii)] There is a unique place $Q$ of degree $2$ of $F$ which is totally ramified in $F/F^{\mG}$ with ramification index $e_Q(F/F^{\mG})=r+1$ and different exponent  $d_Q(F/F^{\mG})=r$. There are exactly $(q+1)/(r+1)$ rational places of $F^{\mG}$ which split completely in  $F/F^{\mG}$.
\item [(iii)] All other places of $F$ are unramified in  $F/F^{\mG}$.
\end{itemize}
\end{prop}
\begin{proof}
From Proposition \ref{prop:2.3}, there is a unique place $Q$ of degree $2$ of $F$ which is totally ramified in $F/F^{\langle \Gs \rangle}$ with ramification index $e_Q(F/F^{\langle \Gs \rangle})=q+1$. Since $\mG$ is a subgroup of $\langle \Gs \rangle$, then $Q$ is totally ramified in
$F/F^{\mG}$ with ramification index $e_Q(F/F^{\mG})=r+1$.
It is easy to see that $r+1$ is relative prime to the characteristic of $\F_q$, i.e., $Q$ is tamely ramified in $F/F^{\mG}$. Hence, the different exponent  $d_Q(F/F^{\mG})=e_Q(F/F^{\mG})-1=r$.

From the Hurwitz genus formula, we have $$2g(F)-2=|\mG| (2g(F^{\mG})-2)+\deg \text{Diff}(F/F^{\mG}).$$
It follows that $\deg \text{Diff}(F/F^{\mG})=2r=\deg(rQ)$. Hence, $Q$ is the unique ramified place in $F/F^{\mG}$.
Then there are $q+1$ rational places of $F$ which are unramified in  $F/F^{\mG}$. Hence, there are exactly $(q+1)/(r+1)$ rational places of $F^{\mG}$ which split completely in  $F/F^{\mG}$.
\end{proof}

Now we can provide an explicit construction of optimal \LRCs with length $n=q+1$ from the subgroups of a cyclic subgroup of order $q+1$ of $\Aut(F/\F_q)$.

\begin{theorem}\label{q+1}
Let $r$ be a positive integer such that $(r+1)|(q+1)$. Put $n=m(r+1)$ for any positive integer $m\le \frac{q+1}{r+1}$. Then for any integer $t$ with  $1\le t\le m$, there exists an optimal $q$-ary $[n,k,d]$-\LRC with locality $r$,  $k=rt$ and $d=n-k-\frac kr+2$.
\end{theorem}
\begin{proof}
We prove the result only for $m=\frac{q+1}{r+1}$. This case gives the largest length $n=q+1$. The reader may refer to Proposition \ref{prop:3.1}  for the proof of the case where $m<\frac{q+1}{r+1}$.

Let $F$ be the rational function field $\F_q(x)$ and let $\mG$ be a subgroup of the cyclic group of order $q+1$ defined in Lemma  \ref{subgroupofq+1} and Proposition \ref{prop:5.2}. In fact, there are exactly $m$ rational places $\{Q_1, \cdots, Q_m\}$ of $F^{\mG}$ which split completely in  $F/F^{\mG}$.
Let $A_j=\{P_{(j-1)(r+1)+1}, \cdots, P_{j(r+1)}\}$ be the set of rational places of $F$  lying over $Q_j$ for  $1\le j\le m$.
Choose an element $z$ of $F^\mG$ such that $(z)_{\infty}=Q_1$ in $F^\mG$. Then we must have $F^\mG=\F_q(z)$.
It is easy to see that $(z)_{\infty}=P_{1}+\cdots+P_{r+1}$ in $F$ and $z$ is a constant function on each $A_j$  since $z(P_{(j-1)(r+1)+\ell})=z(Q_j)$ for every $1\le \ell \le r+1$ and $2\le j\le m$.

Choose an element $y\in F$ such that $(y)_{\infty}=P_{r+2}$ and $F=\F_q(y)$.
For $t\ge 1$, consider the set of functions
\[V:=\left\{\sum_{i=0}^{r-1}\left(\sum_{j=0}^{t-1}a_{ij}z^j\right)y^i:\; a_{ij}\in\F_q\right\}.\]
Let $G=(t-1)(z)_{\infty}+(r-1)(y)_{\infty}=(t-1)(P_{1}+\cdots P_{r+1})+(r-1)P_{r+2}$.
Then $V$ is a subspace of the Riemann-Roch space $\mL(G)$.

Let $m_i=\nu_{P_i}(G)$, then $m_1=m_2=\cdots=m_{r+1}=t-1$, $m_{r+2}=r-1$ and $m_{r+3}=\cdots=m_n=0$. Put $\pi_1=\cdots=\pi_{r+1}=1/z$ and $\pi_{r+2}=1/y$.  Consider the subcode of the modified algebraic geometry code
$$C(\mP,V)=\{((\pi_1^{m_1}f)(P_1),\dots,(\pi_{r+2}^{m_{r+2}}f)(P_{r+2}),f(P_{r+3}),\dots,f(P_n))):\; f\in V\}.$$
It is easy to see that the code $C(\mP,V)$ is an $[n,k,d]$ code with $k=rt$ and $d\ge n-k-\frac kr+2$.

It remains to prove that the code $C(\mP,V)$ has locality $r$. Denote by $(c_1,c_2,\dots,c_n)$ the codeword $((\pi_1^{m_1}f)(P_1),\dots,(\pi_{r+2}^{m_{r+2}}f)(P_{r+2}),f(P_{r+3}),\dots,f(P_n))$ for some $f=\sum_{i=0}^{r-1}$ $\left(\sum_{j=0}^{t-1}a_{ij}z^j\right)y^i\in V$ with $a_{ij}\in\F_q$.
Suppose that the erased symbol of the codeword is $c_{\Ga}=(\pi_{\Ga}^{m_{\Ga}}f)(P_{\Ga})$, where $P_{\Ga}\in A_h$ with $1\le h\le m$.
For $3\le h\le m$, the locality property follows from the proof of Proposition \ref{prop:3.1}.
For $h=1$, define the decoding polynomial $$\Gd(y)=\sum_{i=0}^{r-1} \sum_{j=0}^{t-1}a_{i,j}\pi_1^{t-1} z^j(P_{\Ga}) y^i=\sum_{i=0}^{r-1} a_{i,t-1}y^i.$$
Then as we have shown in Proposition \ref{prop:3.1}, $(\pi_{1}^{t-1}f)(P_{\Ga})=\Gd(y)(P_{\Ga})$ can be recovered from the other $r$ symbols $c_{\Gb}=(\pi_{1}^{t-1}f)(P_{\Gb})$ for $P_{\Gb}\in A_1\setminus \{P_{\Ga}\}$ by the Lagrange interpolation.
For $h=2$, since $\sum_{j=0}^{t-1} a_{ij}z^j$ is constant on $A_2$, we let $b_i=\sum_{j=0}^{t-1} a_{ij}z^j(P_{r+2})$ for $0\le i \le r-2$.  Then we have
$$(c_{r+2}, c_{r+3}, \cdots, c_{2r+2})=\Big{(}b_{r-1}, \sum_{i=0}^{r-1}b_{i}y^i(P_{r+3}), \cdots, \sum_{i=0}^{r-1}b_{i}y^i(P_{2r+2})\Big{)}.$$
{\bf Case 1}: the erased symbol is $c_{r+2}$ at $P_{r+2}$. Since $y(P_{\ell})$ are pairwise distinct for $r+3\le \ell \le 2r+2$, the coefficients $b_i$ for $0\le i \le r-1$ can be determined by the Lagrange interpolation. Hence, the erased symbol $c_{r+2}$ can be recovered.\\
{\bf Case 2}: the erased symbol is $c_{w}$ at $P_{w}$ with  $r+3\le w \le 2r+2$. Then we have
$$\sum_{i=0}^{r-2}b_{i}y^i(P_{\ell})=c_{\ell}-b_{r-1}y^{r-1}(P_{\ell})$$
for each $\ell$ with $r+3\le \ell \neq w \le 2r+2.$
Hence, the coefficients $b_0, \cdots, b_{r-2}$ can be determined by the Lagrange interpolation, i.e., the erased symbol $c_{w}$ can also be recovered from $c_{w}=\sum_{i=0}^{r-1}b_{i}y^i(P_{w}).$
\end{proof}

\begin{ex}\label{ex:5.4}{\rm
For $q=11$, we provide an explicit construction of an optimal [12,2t,14-3t]-\LRC with locality $2$ for $1\le t\le 4.$
First of all, we choose a  quadratic primitive polynomial $x^2-4x+2$ in $\F_{11}[x]$.
Let $\Gs$ the an automorphism of $\F_{11}(x)$ sending $x$ to $1/(-2x+4)$.
Then the order of $\Gs$ is $12$. Let $a_k$ be the numbers defined in the proof of
Lemma \ref{subgroupofq+1}. Hence, they can be computed, i.e., $(a_0,a_1,\cdots,a_{8}, a_{9})=(0,1,4,3,4,-1,-1,-2,-6,2).$
Let $\mG=\langle \Gs^4 \rangle$.  It follows that $F^{\mG}=\F_{11}(z)$, where $$z=x+\Gs^4(x)+\Gs^8(x)=x+\frac{5x+4}{3x-1}+\frac{4x-6}{x+2}=\frac{3x^3+x+3}{3x^2+5x-2}$$
from Equation \eqref{sigma_k}.
The infinity place $\infty$ of $z$, the zero places of $z-2$, $z-3$ and $z-4$ in $F^{\mG}$ are splitting completely into $\{P_{\infty}, P_4, P_9\}$, $\{P_2, P_5, P_6\}$, $\{P_1, P_3, P_{10}\}$ and $\{P_0, P_7, P_8\}$ in $F/F^{\mG}$ respectively.
Let $y=1/x$. Then $(y)_{\infty}=(x)_0=P_0$.
Put \[V:=\left\{\sum_{i=0}^{1}\left(\sum_{j=0}^{t-1}a_{ij}z^j\right)y^i:\; a_{ij}\in\F_{11}\right\}.\]
Then the algebraic geometry code
$$\begin{array}{ll}
C(\mP,V)=\Big{\{} \Big{(}(z^{-3}f)(P_{\infty}),(z^{-3}f)(P_{4}),(z^{-3}f)(P_{9}),(y^{-1}f)(P_{0}), f(P_{7}),f(P_{8}),\\    \quad \ f(P_{2}), f(P_{5}), f(P_{6}), f(P_{1}),f(P_{3}),f(P_{10})\Big{)}:\; f\in V\Big{\}}.
\end{array}$$
is an optimal $[12,2t,14-3t]$-\LRC with locality  $2$ for $1\le t\le 4$.
Furthermore, a generator matrix of an optimal $[12,8,2]$-\LRC can be computed as follows:
$$\left(\begin{array}{cccccccccccc}
0 & 0 & 0 & 0 & 1 & 1 & 1 & 1 & 1 & 1 & 1 & 1 \\
0 & 0 & 0 & 1 & 8 & 7 & 6 & 9 & 2 & 1 & 4 & 10 \\
0 & 0 & 0 & 0 & 4 & 4 & 2 & 2 & 2 & 3 & 3 & 3 \\
0 & 0 & 0 & 4 & 10 & 6 & 1 & 7 & 4 & 3 & 1 & 8 \\
0 & 0 & 0 & 0 & 5 & 5 & 4 & 4 & 4 & 9 & 9 & 9 \\
0 & 0 & 0 & 5 & 7 & 2 & 2 & 3 & 8 & 9 & 3 & 2 \\
1 & 1 & 1 & 0 & 9 & 9 & 8 & 8 & 8 & 5 & 5 & 5 \\
0 & 3 & 5 & 9 & 6 & 8 & 4 & 6 & 5 & 5 & 9 & 6
\end{array}\right).$$
}\end{ex}

\begin{ex} {\rm Let $q=27$. Then by Theorem \ref{q+1}, one can explicitly construct optimal $27$-ary \LRCs with the following parameters
\begin{itemize}
\item[{\rm (i)}]  $[28,t,30-2t]$-\LRCs with locality $1$ for any $1\le t\le 14$;
\item[{\rm (ii)}] $[28,3t,30-4t]$-\LRCs with locality $3$ for any $1\le t\le 7$;
\item[{\rm (iii)}]   $[28,6t,30-7t]$-\LRCs with locality $6$ for any $1\le t\le 4$;
\item[{\rm (iv)}]  $[28,13t,30-14t]$-\LRCs with locality $13$ for any $1\le t\le 2$.
\end{itemize}
}\end{ex}

\begin{ex} {\rm Let $q=64$. Then by Theorem \ref{q+1}, one can explicitly construct optimal $64$-ary \LRCs with the following parameters
\begin{itemize}
\item[{\rm (i)}]  $[65,4t,67-5t]$-\LRCs with locality $4$ for any $1\le t\le 13$;
\item[{\rm (ii)}]  $[65,12t,67-13t]$-\LRCs with locality $12$ for any $1\le t\le 5$.
\end{itemize}
}\end{ex}

\section{Explicit construction via dihedral subgroups}
Let $F$ be the rational function field $\F_q(x)$.
In this section, we give an explicit construction of optimal \LRCs from a subgroup $\mG$ of $\Aut(F/\F_q)$ that is isomorphic to some dihedral subgroup $D_{2h}$ with $h\ge 2$.
For odd $q$, the order $2h$ of $D_{2h}$ is a divisor of $q-1$ or $q+1$. Thus, such localities $r$ have been obtained in Sections 4 and 5 already. In this section, we only consider even $q$. From the Hurwitz genus formula, there are $2+h$ rational places of $F$ which are ramified in $F/F^{\mG}$ (see \cite[Theorem 11.91(iii)]{HKT08}).

\subsection{$h$ is a divisor of $q-1$}
In this subsection, let $h\ge 2$ be a positive divisor  of $q-1$.
Let $\Gs$ be the automorphism of $\F_q(x)$ determined by $\Gs(x)=ax$ for $a\in \F_q^*$ with $\text{ord}(a)=h$ and
let $ \Gt$ be the automorphism of $\F_q(x)$ determined by $\Gt(x)=1/x$.
Let $\mG$ be the subgroup of $\Aut(F/\F_q)$ generated by $\Gs$ and $\Gt$.
Then  it is easy to verify that $\mG=\langle \Gs, \Gt \rangle =\langle \Gs, \Gt| \Gs^h=1, \Gt^2=1,  \Gt \Gs=\Gs^{-1} \Gt \rangle \cong D_{2h}$.

\begin{theorem}\label{D_2m}
Let $h\ge 2$ be a positive divisor  of $q-1$ for even $q$.
Put $r=2h-1$ and  $n=(r+1)m$ for any positive integer $m \le \frac{q-1-h}{2h}$.
Then for any integer $t$ with $1\le t \le m$, there exists an optimal $[n,k,d]_q$-locally repairable code with locality $r$, $k=rt$ and $d=n-k-\frac{k}{r}+2$.
\end{theorem}
\begin{proof} Let $F$ be the rational function field $\F_q(x)$.
Let $\mG=\langle \Gs, \Gt \rangle$  be the subgroup of $\Aut(F/\F_q)$ of order $r+1=2h$ that is isomorphic to the dihedral subgroup $D_{2h}$, where $\Gs, \Gt$ are defined in this subsection.
 Let $P_{\infty}$ be the unique pole place of $x$, i.e., $(x)_{\infty}=P_{\infty}$.
It is easy to verify that $z=x^h+x^{-h}$ is an element of $F^{\mG}$.  Moreover, the pole divisor of $z$ in $F$ is $(z)_{\infty}=hP_{\infty}+hP_0$.  Hence, $F^{\mG}=\F_q(z)$.
For $t\ge 1$, consider the set of functions
\[V:=\left\{\sum_{i=0}^{r-1}\left(\sum_{j=0}^{t-1}a_{ij}z^j\right)x^i:\; a_{ij}\in\F_q\right\}.\]
The rest of the proof is similar as that of Proposition \ref{prop:3.1}.
\end{proof}


\begin{ex}{\rm
Let $q=16$ and $\F_{16}=\F_2(\Ga)$ with $\Ga^4+\Ga+1=0$.
We give an explicit construction of an optimal [12,5t,14-6t]-\LRC with locality $5$ for $1\le t\le 2.$
Let $F$ be the rational function field $\F_{16}(x)$,
let $\Gs(x)=\Ga^5 x$ be an automorphism of $\F_{16}(x)$ of order $3$ and let  $\mG=\langle \Gs, \Gt \rangle$ be the dihedral subgroup $D_6$ of $\Aut(F/\F_{16})$.
It follows that $F^{\mG}=\F_{16}(z)$ where $z=x^3+x^{-3}$.
It is easy to verify that the infinity place $\infty$ of $F^{\mG}$ splits into two places $\{P_{\infty}, P_0\}$ in $F$ with ramification index $3$, the zero place of $z$ of $F^{\mG}$ splits into three places  $\{P_1, P_{\Ga^5}, P_{\Ga^{10}}\}$ in $F$ with ramification index $2$,  the zero places of $z-(\Ga^2+\Ga+1)$ and $z-(\Ga^2+\Ga)$  of $F^{\mG}$  split completely in $F$. In particular, $(z-\Ga^2-\Ga-1)_{0} =P_{\Ga}+P_{\Ga^4}+P_{\Ga^6}+P_{\Ga^9}+P_{\Ga^{11}}+P_{\Ga^{14}}$ and
$(z-\Ga^2-\Ga)_{0} =P_{\Ga^2}+P_{\Ga^3}+P_{\Ga^7}+P_{\Ga^8}+P_{\Ga^{12}}+P_{\Ga^{13}}.$
Put \[V:=\left\{\sum_{i=0}^{4}\left(\sum_{j=0}^{t-1}a_{ij}z^j\right)x^i:\; a_{ij}\in\F_{16}\right\}.\]
Then the algebraic geometry code
$$
\begin{array}{lll}
C(\mP,V)=&\Big{\{} \Big{(} f(P_{\Ga}), f(P_{\Ga^4}), f(P_{\Ga^6}),f(P_{\Ga^9}),f(P_{\Ga^{11}}), f(P_{\Ga^{14}}),\\  & \quad f(P_{\Ga^2}), f(P_{\Ga^3}), f(P_{\Ga^7}),f(P_{\Ga^8}),f(P_{\Ga^{12}}),f(P_{\Ga^{13}})\Big{)}:\; f\in V\Big{\}}.
\end{array}
$$
is an optimal  [12,5t,14-6t]-\LRC with locality $5$ for $1\le t\le 2.$ Similar to Example \ref{ex:5.4}, a generator matrix of such a code can be obtained as well. We skip the details.
}\end{ex}


\begin{ex} {\rm Let $q=64$. Then
\begin{itemize}
\item[{\rm (i)}] by  Theorem  \ref{D_2m}, there exists an optimal $64$-ary $[60,5t,62-6t]$-\LRC with locality $5$ for any $1\le t\le 10$;
\item[{\rm (ii)}] by  Theorem \ref{D_2m}, there exists an optimal $64$-ary $[56,13t,58-14t]$-\LRC with locality $13$ for any $1\le t\le 4$;
\item[{\rm (iii)}] by  Theorem  \ref{D_2m}, there exists an optimal $64$-ary $[54,17t,56-18t]$-\LRC with locality $17$ for any $1\le t\le 3$.
\end{itemize}
}\end{ex}

\subsection{$h$ is a divisor of $q+1$}
Let $h\ge 2$ be a positive divisor of $q+1$ in this subsection.

\begin{theorem} \label{D_2m+}
Let $r$ be a positive integer such that $r+1=2h$, where $h\ge 2$ is a positive divisor $q+1$. Put $n=m(r+1)$ for any positive integer $m\le \frac{q-1-h}{r+1}$. Then for any integer $t$ with  $1\le t\le m$, there exists an optimal $q$-ary $[n,k,d]$-\LRC with locality $r$,  $k=rt$ and $d=n-k-\frac kr+2$.
\end{theorem}
\begin{proof}
Let $F$ be the rational function field $\F_q(x)$ and let $P_{\infty}$ be the infinite place of $F$, i.e., $(x)_{\infty}=P_{\infty}$.
Let $\Gs$ be the automorphism of $\F_q(x)$ of order $q+1$ defined in Lemma \ref{subgroupofq+1}.
Then $\mG:=\langle \Gs^{\frac{q+1}{h}}, \Gt \rangle$ is isomorphic to the dihedral subgroup $D_{2h}$ of $\Aut(F/\F_q)$, where $\Gt$ is given by $\Gt(x)=\frac1x$.
Put $$\mu(x)=\sum_{\rho \in \langle \Gs^{\frac{q+1}{h}} \rangle}\rho(x).$$
It is easy to verify that $z=\mu(x)\cdot\mu(x^{-1})$ is an element of $F^{\mG}$.  Moreover, the degree of the pole divisor of $z$ in $F$ is $\deg (z)_{\infty}=2h$.  Hence, $F^{\mG}=\F_q(z)$.
For $t\ge 1$, consider the set of functions
\[V:=\left\{\sum_{i=0}^{r-1}\left(\sum_{j=0}^{t-1}a_{ij}z^j\right)x^i:\; a_{ij}\in\F_q\right\}.\]
Similarly, the proof can be completed by following
that of Proposition \ref{prop:3.1}.
\end{proof}

\begin{ex} {\rm Let $q=64$. For each divisor $h\ge 2$ of $65$, then
\begin{itemize}
\item[{\rm (i)}] by  Theorem  \ref{D_2m+}, there exists an optimal $64$-ary $[50,9t,52-10t]$-\LRC with locality $9$ for any $1\le t\le 5$;
\item[{\rm (ii)}] by  Theorem \ref{D_2m+}, there exists an optimal $64$-ary $[26,25,2]$-\LRC with locality $25$.
\end{itemize}
}\end{ex}

\section{Conclusion}
In this paper, we show that as long as there is a subgroup $\mG$ of $\PGL_2(q)$ of order $r+1$, one can construct an optimal \LRC with locality $r$. These optimal \LRCs can be explicitly constructed as all subgroups of  $\PGL_2(q)$ have explicit structures. In this paper we provide explicit construction of optimal \LRCs only for several subgroups such as affine subgroups, subgroups of a cyclic group of order $q+1$ and dihedral subgroups. We also provide a general construction by estimating the number of ramified places regardless of subgroup structures.

It is clear that the construction given in this paper can be generalized to arbitrary function fields. We are currently working on the construction of \LRCs from Hermitian function fields and asymptotically optimal towers.

\end{document}